\documentclass[11pt]{article}
\usepackage{cite}      
\usepackage{graphicx}  
\usepackage{caption}

\usepackage{amssymb}
\usepackage{mathrsfs}
\usepackage{amsmath}   
\usepackage{amsthm}

\usepackage[sans]{dsfont}
\usepackage{stmaryrd}
\usepackage{pifont}
\usepackage{wasysym}
\usepackage{algorithm} 
\usepackage{array}
\usepackage{url}
\usepackage{tabu}
\usepackage{multirow}
\usepackage{booktabs}

\usepackage{algpseudocode}
\usepackage{algorithm}

\usepackage{authblk}

\newtheorem{theorem}{Theorem}[section]

\newtheorem{corollary}[theorem]{Corollary}

\newtheorem{definition}[theorem]{Definition}
\newtheorem{example}{Example}

\newcommand{\hide}[1]{}

\title{Economic Viability of Paris Metro Pricing for Digital Services\footnote{This paper appears in ACM Transactions on Internet Technology (ToIT), Special Issue on Pricing and Incentives in Networks and Systems, Vol. 14, No. 12, Issue 2-3, pp12:1-12:21, Oct 2014. A preliminary version has been presented at IEEE INFOCOM 2010 \cite{CWC10pmp}.}}

\author{Chi-Kin Chau}
\affil{Masdar Institute of Science and Technology}
\author{Qian Wang}
\affil{Alibaba, Inc.}
\author{Dah-Ming Chiu}
\affil{Chinese University of Hong Kong}

\begin{document}

\maketitle

\begin{abstract}
Nowadays digital services, such as cloud computing and network access services, allow dynamic resource allocation and virtual resource isolation. This trend can create a new paradigm of flexible pricing schemes. A simple pricing scheme is to allocate multiple isolated service classes with differentiated prices, namely Paris Metro Pricing (PMP). The benefits of PMP are its simplicity and applicability to a wide variety of general digital services, without considering specific performance guarantees for different service classes. The central issue of our study is whether PMP is economically viable, namely whether it will produce more profit for the service provider and whether it will achieve more social welfare. Prior studies had only considered specific models and arrived at conflicting conclusions. In this article, we identify unifying principles in a general setting and derive general sufficient conditions that can guarantee the viability of PMP. We further apply the results to analyze various examples of digital services.    
\end{abstract}


\section{Introduction} \label{sec:intro} 

The management of digital services, such as cloud computing, video streaming and gaming services, and network access services in wireline and wireless networks, is increasingly dictated by economic principles. Digital services possess several salient characteristics that allow more flexible resource management and allocation mechanisms that pave the way for a variety of innovative pricing schemes. 

First, there emerge new technologies for virtual resource isolation, by which resources can be allocated conveniently (either dynamically or a priori) to create isolated service classes without altering the underlying hardware infrastructure. Virtual resource isolation can be realized by a variety of technologies. For example, computational virtualization can create virtual machines (VMs) to abstract the underlying hardware resources from applications. Each VM can execute isolated computing tasks without affecting others. On the other hand, software-defined radio has been utilized at base stations to enable dynamic spectrum allocation for different standards in wireless access networks, whereas multiprotocol label switching (MPLS) can assign different routing and queuing operations for packets depending on specific applications. With virtual resource isolation, multiple service classes can be conveniently established. 

Second, the availability of real-time measurement of services enables users to observe the performance of other service classes. Users may make instantaneous decisions and adjustment to their behavior according to the observed performance. In the presence of multiple service classes, users who cannot tolerate the performance in the respective class may choose to switch between classes as a consequence. Third, pricing models are evolving so as to be more flexible. The rise of on-demand pricing models, such as pay-per-usage, allows users to readily opt out of a service. 

Therefore, it is natural to consider pricing schemes for multiple service classes in digital services. Nowadays, differentiated pricing with multiple service classes has been observed in practice. For example, streaming and file sharing service providers offer premier and economic service classes that can run simultaneously on the same cloud computing platform with isolated VMs and partitioned bandwidth. In mobile wireless service, 3G and 4G can potentially share a spectrum allocated dynamically at base stations. A simple pricing scheme is to impose differentiated prices at different service classes without fulfilling specific quality-of-service requirements and let users spontaneously opt for the appropriate service classes according to their experienced performance. For example, a premier service class charged at a higher price will attract fewer low-end users and therefore is able to provide superior performance to the high-end users. This gives rise to so-called Paris Metro Pricing (PMP), a simple multiclass flat-rate pricing scheme.

PMP is an attractive approach due to its simplicity. However, a more subtle question is whether PMP is economically viable and specifically, {\em whether we can improve the profit and social welfare through a suitable pricing scheme on service classes with an appropriately allocated amount of resources}. Prior studies arrived at conflicting conclusions for this question. On the one hand, \cite{gibbens00JSAC,jain01pmp} found PMP to be viable, while on the other, based on a similar model, \cite{rt04pmp} numerically showed that PMP may not be more viable than flat-rate pricing. 
 
In this article, we observe unifying principles that depend on the nature of the externality of the underlying service models. \cite{gibbens00JSAC,jain01pmp} assumed one type of congestion function, whereas \cite{rt04pmp} assumed another, hence reaching conflicting conclusions. We consider a general model that can capture a wide variety of digital services and we provide general sufficient conditions for the viability of PMP in terms of both social welfare as well as provider profit. This leads to the insights on \emph{why} PMP is or is not viable, understandable by common practitioners. We further apply our results to analyze various examples of digital services.   


{\bf Outline}: Sec.~\ref{sec:relate} provides the background and related work. Sec.~\ref{sec:model} formulates a general model of PMP, and gives several examples of digital services that it captures. Sec.~\ref{sec:res} gives an overview of our analytic results considering monopoly. We then extend our study to duopoly, supported by numerical studies in Sec.~\ref{sec:duo}. 

\section{Background and Related Work} \label{sec:relate} 

PMP was first proposed by Odlyzko \cite{odlyzko98pmp} as a simple pricing model for the Internet-differentiated services and can better satisfy users with different aversion to effects of congestion as divided by different service classes. 
The scheme is inspired by the convention used by Paris metro at one time\footnote{This scheme is actually adopted rather widely in other transportation systems in the world, including the Mass -Transit Railway in Hong Kong.}: The first- and second-class cars are charged with different prices, although physically the cars are the same (in terms of the number and quality of the seats). Since fewer people would pay more for the first-class fare, it is also less congested. Thus, users more concerned about getting a seat can opt for first class, and more cost conscious users can opt for second. Note that, in each class, the user is still paying a flat rate. PMP has a self-stabilizing property, namely that if the performance of first class deteriorates, some users will switch to second class, thus increasing the quality differential between the classes.

In general, the main types of pricing schemes adopted in practice include\footnote{{One can also apply a static (time-invariant) pricing strategy as well as a dynamic pricing strategy (that is contingent on the history) to these pricing schemes. Also, the pricing strategy may be application-specific \cite{sen12survey}.}}: (1) {\em flat-rate pricing}, that is, to charge a one-off payment for every user, regardless her usage. (2) {\em usage-based pricing}, that is, to charge according to the amount and pattern of usage of each individual user. In this broad classification, we also include congestion pricing as a form of usage-based pricing. In the presence of multiple service classes, one can apply flat-rate or usage-based pricing to each of the service classes.
Congestion pricing, levied only when resource demand exceeds supply, can be argued as economically and theoretically the most optimal strategy for allocating congested resources. In practice, however users strongly prefer flat-rate pricing \cite{odlyzko12flat}\footnote{
For example, while the Internet service provider (ISP) settlements on aggregate transited traffic are typically based on usage-based pricing, ISPs charge users using flat-rate pricing for its simplicity. This observation also applies to other digital services. While flat-rate pricing is simple, it is insufficient to control the desirable performance to satisfy users' utility. Users are intolerant of inferior performance are forced to opt out. 
On the other hand, usage-based pricing enables more sophisticated control of performance, but incurs a higher implementation cost because of its more intrusive monitoring and policing on the usage pattern. }.

Since PMP relies on the spontaneous economics adjustment of users, the outcome may not align with the goals of mechanism designers. There have been studies \cite{gibbens00JSAC,steinberg00duo,jain01pmp,sakurai03optout,rt04pmp,SSOA07pmp} in the literature to address the viability of PMP that consider the following two major criteria: (1) {\em Social welfare}, the total utility of users, taking into account the congestion. (2) {\em Provider profit}, the total payment collected by the service providers from different service classes. 

Based on a multiproduct economics model in \cite{chander89mono,palma89duopoly}, \cite{gibbens00JSAC} formulates a model of PMP, considering the massive number of infinitesimal users and analyzes the viability of PMP for social welfare and provider profit. However, \cite{gibbens00JSAC} consider only a specific model (that we call utilization-sensitive service) and assume no user will opt out\footnote{Allowing user opt-out leads to a variable number of users in the system that is equivalent to elastic demand, whereas inelastic demand is equivalent to a constant fixed number of users in the system.}.  Next, \cite{steinberg00duo} consider a less specified model but focusing on provider profit and still under the assumption of no user opt-out and homogeneous service classes (where the resources allocated to each class are identical). 

On the other hand, \cite{jain01pmp,sakurai03optout} consider utilization-sensitive service with user opt-out, but only study provider profit. Meanwhile, \cite{rt04pmp} numerically studies PMP under a different model (which we call latency-sensitive service) for only provider profit. 
In the prior work of utilization-sensitive service (\cite{gibbens00JSAC,jain01pmp,sakurai03optout}), it is reported that PMP is viable for a single monopoly provider. Nonetheless, \cite{rt04pmp} reports that PMP is not viable in latency-sensitive service for a monopoly provider based on only numerical analysis. We observe that different models can give contradictory results. But there are yet any studies to provide a complete picture considering a general model, beyond the specific examples of utilization- and latency-sensitive services. 

To provide a unifying picture, we first analytically study the case of single monopoly provider. We then extend to a duopoly setting for different models. In contrast to the observation in \cite{gibbens00JSAC}, we find that PMP is also viable for duopoly for a certain model of congestion with user opt-out. Our analysis is supported by extensive numerical study.
Our study follows the popular model of {\em infinitesimal users}, as in the prior work (\cite{gibbens00JSAC,steinberg00duo,jain01pmp,sakurai03optout,rt04pmp}). We note that a different model of PMP with finitely many users has been studied in \cite{SSOA07pmp} showing that in a specific setting of user utility function wherein a single service class is strictly better than multiple service classes for a monopoly provider. Their conclusion agrees with \cite{rt04pmp} and generally with ours, albeit by a different model with finitely many users. 


\section{Model and Notations} \label{sec:model}

This section presents a model for PMP based on different possible forms of negative externality, generalizing the models from \cite{gibbens00JSAC,steinberg00duo,jain01pmp,sakurai03optout,rt04pmp}, and defines the notations of \emph{equilibrium} (in which users settle their selections of service classes) and \emph{social welfare} and \emph{provider profit}. 

Regardless of the technologies of resource isolation in digital services, it is vital to understand the economical viability of such a service model. Without specifying the implementation details of resource isolation, we assume that resources can be conveniently split or merged among different service classes. 

\subsection{Utility and Services Classes}

Suppose that there are $m$ service classes. Similar to many economic studies \cite{gibbens00JSAC,jain01pmp,sakurai03optout}, we assume there are a large number of users. This can be approximated by a continuum model of {\em infinitesimal users} such that the type of each infinitesimal user is characterized by a one-dimensional valuation of a positive real parameter $\theta$. The consideration of infinitesimal user is widely used in economics literature that concerns massive numbers of users, such as in digital services. One expects that the more detailed model of finitely many users will approach the simpler one of infinitesimal users when the number of users becomes large.

When a user of type $\theta$ (in short, we call user $\theta$) uses service class $i \in \{1, ..., m\}$, we assume that its utility is given by: 
\begin{equation} \label{eqn:usr_util}
U_\theta(i) \triangleq V - p_i - \theta \cdot K(Q_i, C_i) 
\end{equation}
These parameters are explained as follows.
\begin{enumerate}

\item $V$ is the maximum utility of accessing the service.

\item $p_i$ is the one-off payment charged per user when accessing the $i$-th service class. Without loss of generality, we assume $V \ge p_1 \ge p_2 \ge \dots \ge p_m \ge 0$.

\item $C_i$ is the proportion of total capacity of the $i$-th service class, such that $\sum_{i=1}^{m} C_i = 1$.

\item $Q_i$ is the volume of users of accessing the $i$-th service class.

\item $K$ is a congestion function that is increasing in $Q_i$, but is decreasing in $C_i$, to be explained in the next section.

\end{enumerate}

Therefore, user $\theta$ will have two options either: (1) to select the $i$-th service class to join that gives the highest utility as  
\begin{equation} \label{eqn:join_iclass} 
i = {\arg \max}_{j \in \{1,...,m \}} U_\theta(j),  
\end{equation}
or (2) to opt out of all service classes because joining any service class will result in a negative utility, that is, $U_\theta(i) < 0$ for all $i \in \{1,...,m\}$.

An immediate implication of utility function $U_\theta(i)$ (Eqn.~(\ref{eqn:usr_util})) is that, the larger the value of $\theta$ means the higher the valuation on the negative externality, as compared with the price of each service class. Hence, users with the larger value of $\theta$ will be more likely to opt out of the service because of negative utility $U_\theta(i) < 0$ for all $i$.

{\bf Remark:} 
Utility function can be alternately defined as $\tilde{U}_\theta(i) \triangleq \theta \cdot \tilde{K}(Q_i, C_i) - p_i$, where $\tilde{K}(Q_i, C_i)$ represents a satisfaction function capturing the inverse effect of a congestion function\footnote{In this case, users with smaller value of $\theta$ will be likely to opt out of the service because of negative utility. Previously, $\tilde{U}_\theta(i)$ was used in the literature of economics \cite{chander89mono,palma89duopoly}.}. We consider utility function as Eqn.~(\ref{eqn:usr_util}) in this article, because there are special cases used in the prior work in networking research community \cite{gibbens00JSAC,steinberg00duo,jain01pmp,sakurai03optout,rt04pmp}. Nonetheless, the results derived in this article can be easily adapted and similarly applied to $\tilde{U}_\theta(i)$. Also, we remark that \cite{chander89mono,palma89duopoly} consider only homogeneous service classes (i.e., $C_i$ is the same in all $m$ service classes). We relax this constraint to include heterogeneous service classes of different $C_i$'s.

\subsection{Congestion Functions of Digital Service Models} \label{sec:example}

Under PMP, users first observe the performance of different service classes and then opt for the appropriate classes. The performance of digital services is characterized by {\em negative externality}, where, the greater number of users accessing a certain service, the less favorable performance the users can perceive. The degree of negative externality also depends on the resource allocated to the service. Hence, negative externality can be characterized by: (1) the volume of users accessing the service, and (2) the amount of allocated resource. These two quantities form the basis of negative externality.

In this article, we employ a congestion function $K(Q,C)$ to capture the negative externality, where $C$ represents the numerical amount of resource and where $Q$ represents the volume of users accessing the service. We normalize $Q$, such that $0 \le Q \le C$.
Table~\ref{tab:congf} provides several examples of $K(Q,C)$ that can model various metrics of negative externality in common digital services.

\begin{table}\centering
\caption{Examples of Metrics of Negative Externality for Digital Services\label{tab:congf}}{
\begin{tabular}{c@{ \quad }c@{ \quad }c@{ \quad }c}
\hline\hline 
& & & \\
Utilization & Latency & Loss Probability & Outage Probability \\ 
& & & \\
\hline
& & & \\
$K_{\sf utl}(Q,C) \triangleq\frac{Q}{C}$ & $K_{\sf lat}(Q,C) \triangleq\frac{1}{C - Q}$ & $K_{\sf los}(Q,C) \triangleq(\frac{Q}{C})^k \frac{1- \frac{Q}{C}}{1- (\frac{Q}{C})^{k+1}}$ & $K_{\sf out}(Q,C) \triangleq (\frac{\epsilon Q}{C})^C$ \\
& & & \\
\hline\hline
\end{tabular}}
\end{table} 

\begin{enumerate}

\item {\em Utilization}: This is measured by the portion of capacity per each unit of load, given by $K_{\sf utl}(Q,C) \triangleq \frac{Q}{C}$. Utilization is a useful metric for computation, bandwidth, and memory sharing services, such as cloud computing and network access services. $K_{\sf utl}(Q,C)$ was considered in prior work \cite{gibbens00JSAC}.

\item {\em Latency}: This is conveniently captured by a simple M/M/1 queue. Assuming the arrival rate is of $Q$ units, and the service rate is of $C$ units, then the total expected waiting time (i.e., queuing time plus service time) is $K_{\sf lat}(Q,C) \triangleq \frac{1}{C-Q}$. $K_{\sf lat}(Q,C)$ was considered in prior work \cite{rt04pmp}. Latency is a useful metric for queuing sensitive services such as streaming or network routing services. We also consider more general M/G/1 queue in the later section.

\item {\em Loss Probability}: We also consider M/M/1/$k$ queue to model queuing-based digital services. An important metric is whether a request will be dropped when all $k$ servers are occupied. In the M/M/1/$k$ model, the probability that $k$ servers are occupied is given by $K_{\sf los}(Q,C) \triangleq (\frac{Q}{C})^k \frac{1- \frac{Q}{C}}{1- (\frac{Q}{C})^{k+1}}$.

\item {\em Outage Probability}: To model the reliability of digital services, we consider a small probability that a server will fail. A natural setting of failure probability is proportional to the utilization $\frac{\epsilon Q}{C}$. If there are $C$ servers, then the probability that all servers will fail is given by $K_{\sf out}(Q,C) \triangleq (\frac{\epsilon Q}{C})^C$.

\end{enumerate}

{\bf Observations:}
Different congestion functions lead to different results. As a clear illustration, we consider a simple scenario of resource partitioning into two identical service classes charged at an identical price such that each class is allocated with half the original resource. Because of identical resource and price, it is likely the usage will be split equally between the two identical service classes. Hence, we define the resource and usage for each service class by $C_1 = C_2 = \frac{C}{2}$ and $Q_1 = Q_2 = \frac{Q}{2}$, where the subscript indicates the first or second class, and $C, Q$ are the resource and usage of the original nonsplit service class.  
For utilization-sensitive service, the congestion functions after service partitioning become 
\begin{equation} 
K_{\sf utl}(Q_1,C_1) = K_{\sf utl}(Q_2,C_2) = \frac{Q}{C} = K_{\sf utl}(Q,C)  
\end{equation} 
In other words, the users will not perceive any difference in terms of negative externality. 
On the other hand, for latency-sensitive service, the congestion functions after service partitioning become 
\begin{equation} 
K_{\sf lat}(Q_1,C_1) = K_{\sf lat}(Q_2,C_2) = \frac{2}{C-Q} > K_{\sf lat}(Q,C)  
\end{equation}
In this case, the users actually perceive a degradation of service after service partitioning!
Therefore, we may conclude that the provision of multiple service classes provided by resource partitioning is not viable for latency-sensitive service, because of the decrease of total welfare of users and also the decrease of provider profit, as some users may opt out of the service. 
The preceding simple discussion, however, has not taken into consideration that different service classes may have varying prices. A more complete study requires a model including payment as a part of the user utility and a clear description of the users' decision process for switching between service classes and opting out. 
This will be completed in the following sections.

\subsection{Equilibrium, Social Welfare and Provider Profit}

In this article, we will analyze the viability of PMP at equilibrium. An {\em equilibrium} will be attained when no user switches from his selection. To formulate equilibrium, we first note that $K(Q_i, C_i)$ is fixed for each $i$-th service class at equilibrium, thus also the price $p_i$. Hence, the utility function $U_\theta(i)$ (Eqn.~(\ref{eqn:usr_util})) becomes a linear function of $\theta$ with $-K(Q_i, C_i)$ as the slope and $V-p_i$ as the $y$-intercept. For the case $m=2$, we plot the utility $U_\theta(i)$ against $\theta$ in Fig.~\ref{fig:usr_equ} for illustration.

\begin{figure}[htb] 
\centering 
  \includegraphics[scale=0.5]{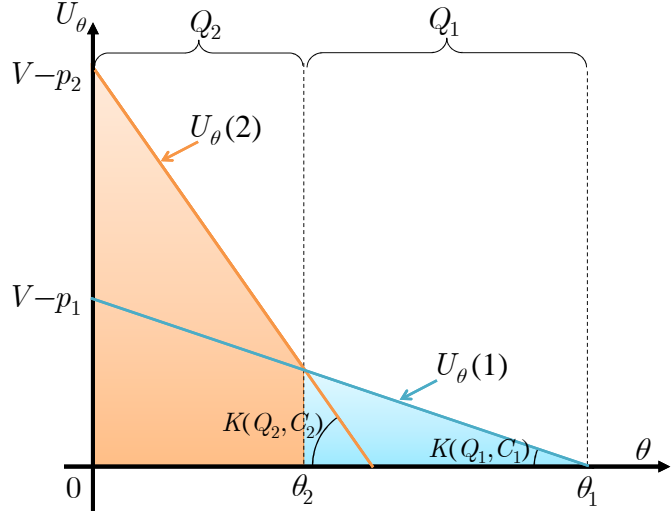}
  \caption{An illustration of equilibrium. We plot the utility $U_\theta(i)$ against $\theta$ for the case $m=2$. We assume $F(\theta) = \theta$ is a uniform distribution.} 
  \label{fig:usr_equ} 
\end{figure}

Given a vector of differentiated prices, denoted by $\mbox{\boldmath $p$} = (p_i)_{i=1}^{m}$, for the $m$ service classes, there exists a set of {\em cut-off users}, denoted by $\mbox{\boldmath $\theta$} = (\theta_i)_{i=1}^{m}$, such that, for $i = 2,..,m$, each cut-off user $\theta_i$ is indifferent to joining the $(i-1)$-th service class or the $i$-th service class (i.e., $U_{\theta_i}(i-1) = U_{\theta_i}(i)$), and cut-off user $\theta_1$ is indifferent to joining the first service class or opting out of these $m$ service classes. Hence, an equilibrium can be characterized by the tuple $(\mbox{\boldmath $p$}, \mbox{\boldmath $\theta$})$. 

Let $Q_0 \triangleq 1 - \sum_{i=1}^{m}(Q_i)$ be the volume of users
who opt out of these $m$ service classes. The users are parameterized by $\theta$, which is described by a cumulative distribution function $F(\theta)$ (and its probability density function denoted by $f(\theta)$). A special
case is that $F(\theta)$ is a uniform distribution.

\begin{definition}({\em Equilibrium}) \label{def:equ}
The tuple $(\mbox{\boldmath $p$}, \mbox{\boldmath $\theta$})$ defines an equilibrium, if the following constraints are satisfied:

\begin{enumerate}
\item[(${\sf c.1}$):] $\theta_1  > \theta_2  >  \dots  > \theta_m  > \theta_{m+1} = 0$,

\item[(${\sf c.2}$):] $ K(Q_i, C_i) \le K(Q_{i+1}, C_{i+1})$ for $i \ne m$, where 
\begin{equation} \label{eqn:c2}
Q_i \triangleq \left\{
\begin{array}{rl}
F(\theta_m) & \mbox{if } i = m \\
F(\theta_{i}) - F(\theta_{i+1}) & \mbox{if } 1 \le i < m
\end{array}
\right. 
\end{equation}

\item[(${\sf c.3}$):] \quad 
\begin{equation} \label{eqn:c3} \hspace{-25pt} 
\begin{array}{@{}r@{}l@{}l} 
p_{i-1} - p_{i} & = \theta_{i} \cdot \big( K(Q_{i}, C_{i}) - K(Q_{i-1}, C_{i-1})\big) & \mbox{\ if\ } 1 < i \le m \\
p_1 & = V - \theta_1 \cdot K(Q_1, C_1) & \mbox{\ otherwise\ } 
\end{array} 
\end{equation}

\end{enumerate}

\end{definition}

In order words, (${\sf c.1}$) requires the set of cut-off users $(\theta_i)_{i=1}^{m+1}$ to have a strict order of valuations on negative externality. User $\theta_{m+1} = 0$ is the least-valuation user who always accepts the lowest-priced service class (i.e., the $m$-th service class). (${\sf c.2}$) follows from $V \ge p_1 \ge \dots \ge p_m \ge 0$ and the definition of utility function (Eqn.~(\ref{eqn:usr_util})). Lastly, (${\sf c.3}$) characterizes the prices of service classes based on the fact that each cut-off user $\theta_i$ is indifferent to joining the $(i-1)$ and $i$-th service classes (i.e., $U_{\theta_i}(i-1) = U_{\theta_i}(i)$). These conditions are depicted in Fig.~\ref{fig:usr_equ}.

Note that it is possible that $p_i = p_{i-1}$ (i.e., $K(Q_i, C_i) = K(Q_{i-1}, C_{i-1})$) for some $i$. It is easy to see that an equilibrium always exists given either \mbox{\boldmath $p$} or \mbox{\boldmath $\theta$}. Namely, there is a one-to-one mapping between \mbox{\boldmath $p$} and \mbox{\boldmath $\theta$} at equilibrium (see \cite{chander89mono} for a rigorous proof). 

To facilitate the analysis, we also assume that $F(\theta)$ is a well-formed distribution, such that $f(\theta) > 0$ for $\theta \in [0, \bar{\theta}] \subseteq [0,1]$ for constant $\bar{\theta}$, ands $f(\theta) = 0$ otherwise. This assumption prevents discontinuity of the marginal change of \mbox{\boldmath $p$} with respect to the marginal change of \mbox{\boldmath $\theta$} at equilibrium.

Finally, we define the social welfare as the total utility of all users excluding the payment and define the monopoly provider profit as the total payment collected from the users.
Suppose $(\mbox{\boldmath $p$}, \mbox{\boldmath $\theta$})$ is an equilibrium. Let the {\em social welfare} be 
\begin{equation} \label{eqn:soc_wel}
  S(\mbox{\boldmath $p$}) \triangleq \sum_{i=1}^{m} \int_{\theta_{i+1}}^{\theta_{i}}
\Big( V-  \theta  \cdot K(Q_i, C_i) \Big) \cdot f(\theta ){\sf d}\theta 
\end{equation}
Note that in Eqn.~(\ref{eqn:soc_wel}), $\mbox{\boldmath $\theta$}$ is a function of $\mbox{\boldmath $p$}$ at equilibrium. Let the monopoly {\em provider profit} be  
\begin{equation} \label{eqn:mono_prof}
\pi(\mbox{\boldmath $p$}) \triangleq \sum_{i=1}^{m}  p_i  \cdot Q_i 
\end{equation}

\subsection{Numerical Study} \label{sec:numerical}

Before presenting our analytical results, this section provides some illustrations of the consequences of PMP in terms of social welfare and provider profit at equilibrium by means of numerical studies. We consider four specific examples, namely utilization-, latency-, loss- and outage-sensitive services. Specifically, we compare: (1) the case of one single service class ($C = 1$), and (2) the case of two service classes ($C_1 = 0.3, C_2 = 0.7$). In the numerical example, we let the probability distribution $F(\theta)$ be a uniform distribution and the maximum utility be $V = 2$.

We compare the maximum social welfare and provider profit that can be achieved between one single service class and two service classes. For two service classes, we let $p_2 = a \cdot p_1$, where $a \in [0, 1]$. When $a = 1$, it is equivalent to identical pricing at two service classes. We numerically evaluated the following four quantities for different service models. 
\begin{equation}
\max_{0 \le p \le V} \pi(p), \qquad  
\max_{0 \le p_1 \le V} \pi(p_1, a \cdot p_1), \qquad  
\max_{0 \le p \le V} S(p), \qquad  
\max_{0 \le p_1 \le V} S(p_1, a \cdot p_1)  
\end{equation}

Figs.~\ref{fig:cs_welfare}-\ref{fig:lb_profit} show the respective numerical values of maximum social welfare and provider profit at given values of $a$, for utilization-, latency-, loss- and outage-sensitive services. Note that, when there is one service class, the maximum social welfare and provider profit are constant values with respect to $a$.

\begin{figure*}[htb!] 
\centering 
  \begin{minipage}[c]{.47\textwidth}
    \includegraphics[scale=0.65]{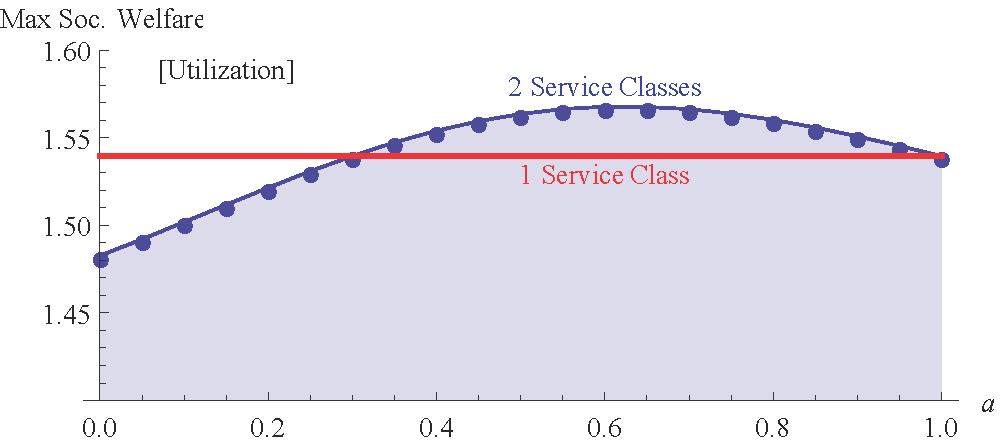}  
    \caption{Maximum social welfare for utilization-sensitive service, where $p_2 = a \cdot p_1$.} \label{fig:cs_welfare}
  \end{minipage}
  \hfill \quad
  \begin{minipage}[c]{.47\textwidth}
    \includegraphics[scale=0.65]{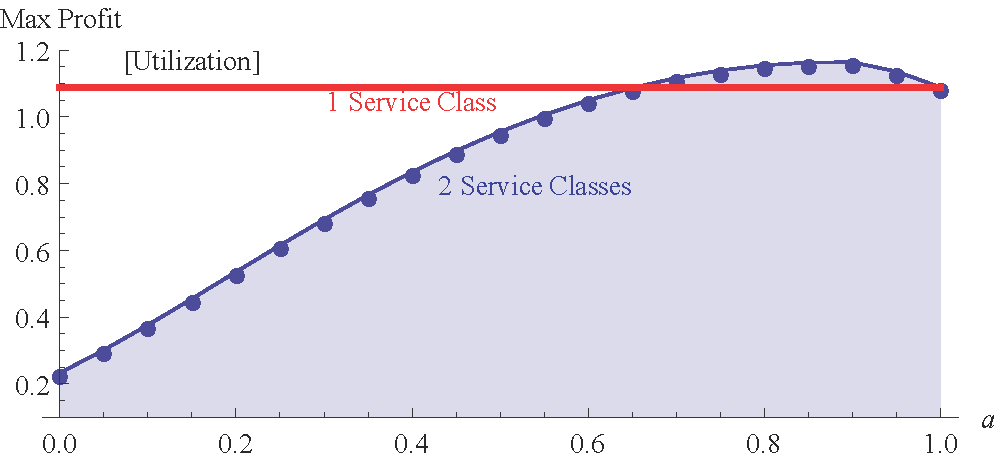} 
    \caption{Maximum provider profit for utilization-sensitive service, where $p_2 = a \cdot p_1$.} \label{fig:cs_profit}
  \end{minipage} 
  \begin{minipage}[c]{.47\textwidth}
    \includegraphics[scale=0.65]{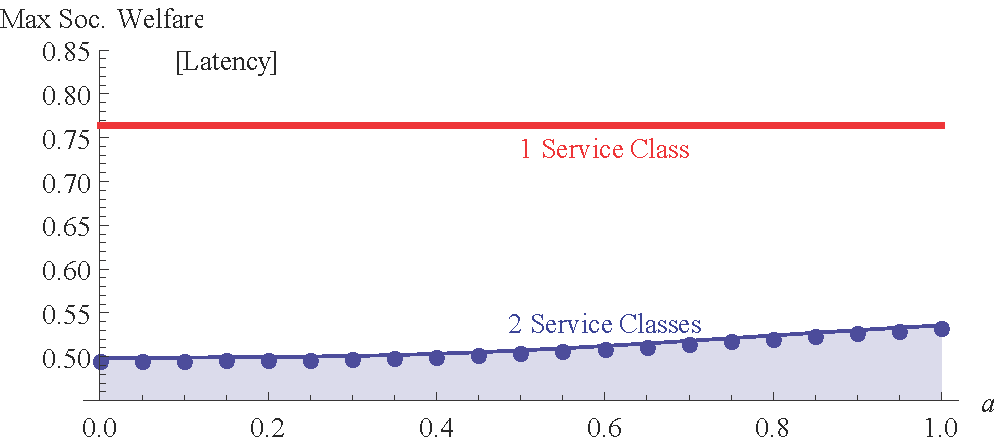}  
    \caption{Maximum social welfare for latency-sensitive service, where $p_2 = a \cdot p_1$.} \label{fig:lb_welfare}
  \end{minipage}
  \hfill \quad
  \begin{minipage}[c]{.47\textwidth}
    \includegraphics[scale=0.65]{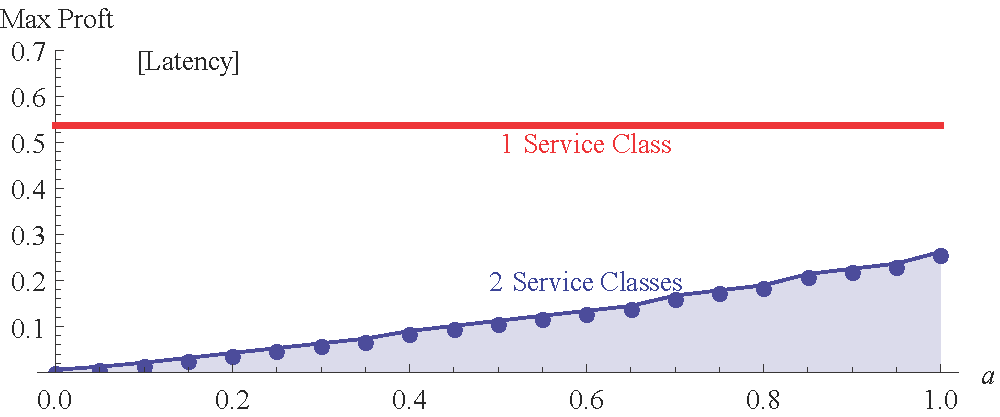}  
    \caption{Maximum provider profit for latency-sensitive service, where $p_2 = a \cdot p_1$.} \label{fig:lb_profit}
  \end{minipage} 
  \begin{minipage}[c]{.47\textwidth}
    \includegraphics[scale=0.65]{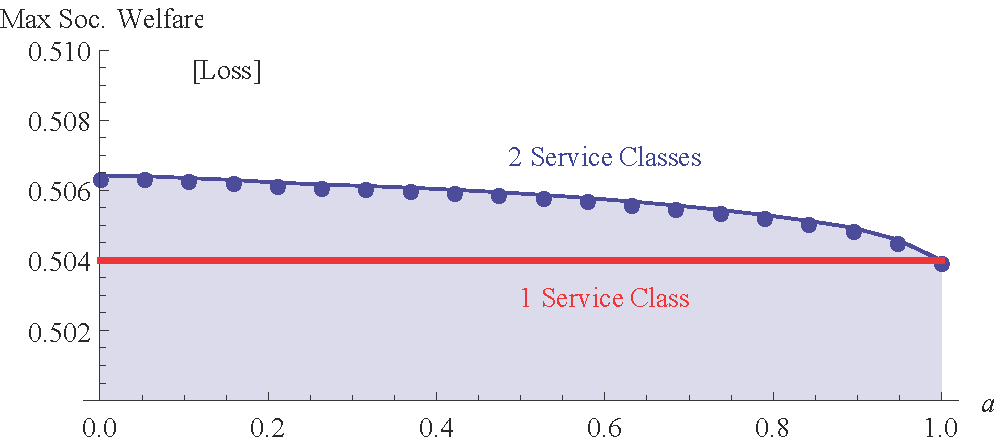}  
    \caption{Maximum social welfare  for loss-sensitive service, where $p_2 = a \cdot p_1$.} \label{fig:lo_welfare}
  \end{minipage}
  \hfill \quad
  \begin{minipage}[c]{.47\textwidth}
    \includegraphics[scale=0.65]{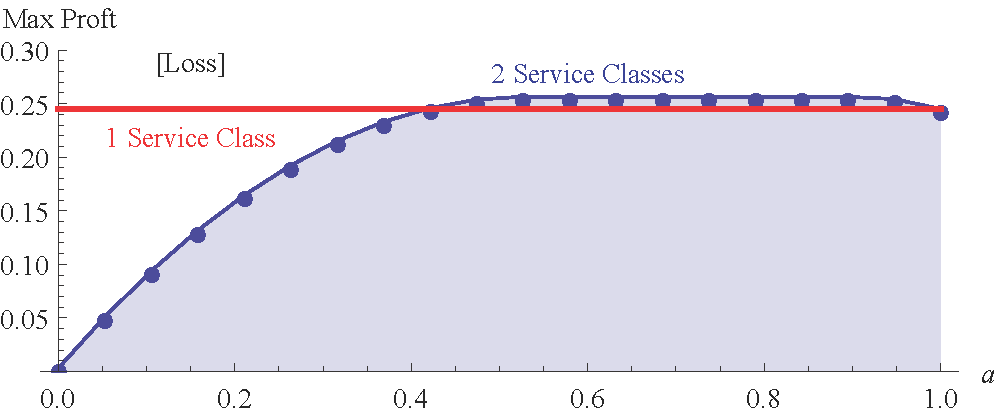}  
    \caption{Maximum provider profit for loss-sensitive service, where $p_2 = a \cdot p_1$.} \label{fig:lo_profit}
  \end{minipage} 
  \begin{minipage}[c]{.47\textwidth}
    \includegraphics[scale=0.65]{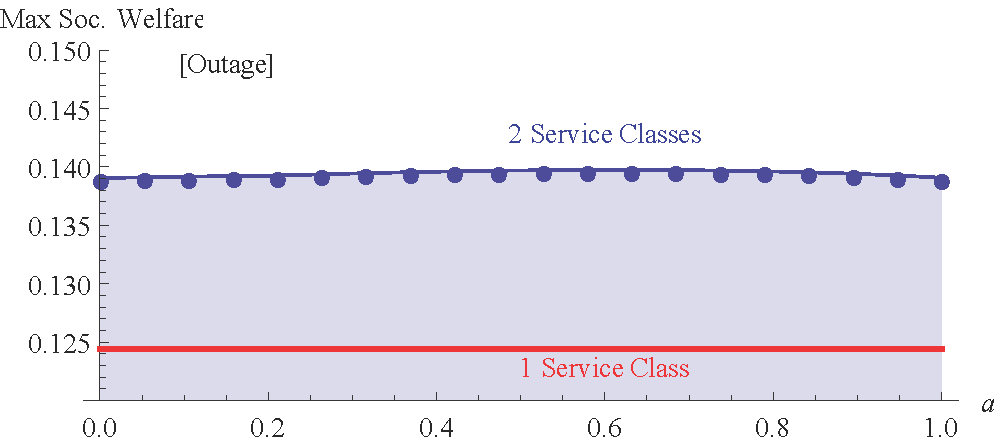}  
    \caption{Maximum social welfare for outage-sensitive service, where $p_2 = a \cdot p_1$.} \label{fig:ou_welfare}
  \end{minipage}
  \hfill \quad
  \begin{minipage}[c]{.47\textwidth}
    \includegraphics[scale=0.65]{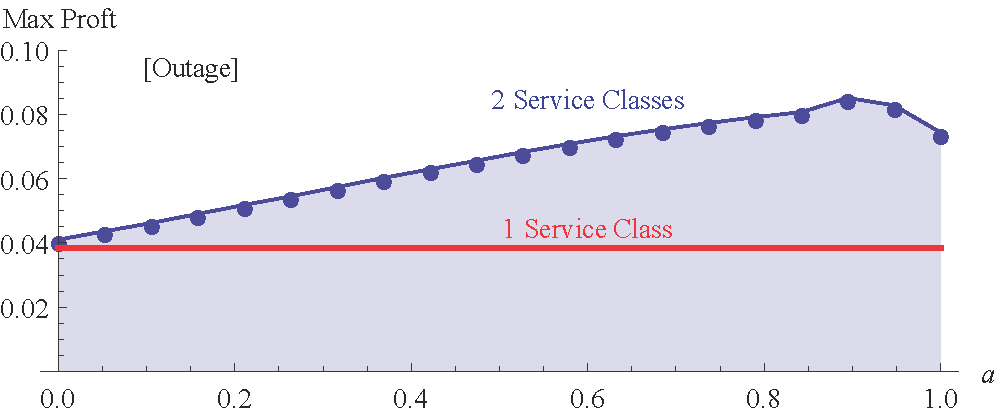}  
    \caption{Maximum provider profit for outage-sensitive service, where $p_2 = a \cdot p_1$.} \label{fig:ou_profit}
  \end{minipage} 
\end{figure*}

{\bf Observations:}
We obtain the following observations from Figs.~\ref{fig:cs_welfare}-\ref{fig:ou_profit} that motivate our later results.
\begin{enumerate}

\item For utilization- and loss-sensitive services, the social welfare and provider profit for one service class and two service classes under identical pricing (i.e., $a=1$) are equivalent. But for latency- and outage-sensitive service, there is a gap between the cases, that is, one service class yields better social welfare and profit for latency-sensitive service, whereas it is worse for outage-sensitive service.

\item For utilization-, loss- and outage-sensitive services, the service provider can make more profit from PMP under a certain setting of differentiated pricing (see Figs~\ref{fig:cs_profit}, \ref{fig:lo_profit}, \ref{fig:ou_profit}); but for latency-sensitive service, PMP (whether under differentiated pricing or identical pricing) is not viable from the provider profit's point of view (see Fig~\ref{fig:lb_profit}). These are exactly the contradictory conclusions reached by \cite{jain01pmp} and \cite{rt04pmp}, respectively. We observe the same conclusions for social welfare. 

\item We remark that, in general, for latency-sensitive service with two service classes, identical pricing may or may not provide higher social welfare or profit than differentiated pricing. However, for utilization-, loss- and outage-sensitive services with two service classes, it is always possible that differentiated pricing can provide higher social welfare and profit than identical pricing.

\end{enumerate}
The numerical results just given motivate the goal for this study, namely how to settle the question of when PMP can be guaranteed to yield higher profits and achieve more social welfare. We answer this question by deriving conditions on a general class of congestion functions, beyond just the cases of utilization- and latency-sensitive services. 

Tactically, we study this problem in two steps by asking the following questions.  
\begin{enumerate}

\item How do we ensure the viability of partitioning a service class into multiple service classes under identical pricing? 

\item How do we ensure the viability of differentiated pricing of multiple service classes as compared with identical pricing?

\end{enumerate}
Resolving these two questions can lead to a characterization of the viability of PMP. Our answers to these questions are elaborated in Sec.~\ref{sec:res}.

\section{Monopoly Case} \label{sec:res}

The viability of PMP ultimately depends on the basic properties of the congestion function. In this section, we provide a general analytical study by identifying some key properties of congestion functions associated with the viability of PMP considering monopoly. 

\subsection{Viability of Service Partitioning} \label{sec:partition}

We first consider identical pricing and provide insights for the viability of {\em service partitioning}, wherein one single service class is partitioned into two service classes, both of which are priced the same as the original service class and, each service class is allocated a portion of the capacity of the original service class. 
Initially, there is one single service class. Let $p$ and $C$ be the price and capacity of the original single service class. Let $\tilde{\theta}$ be the cut-off user at equilibrium (who is indifferent to either opting-out or joining the service) and the total usage be $\tilde{Q} \triangleq F(\tilde{\theta})$. Then, the social welfare becomes 
\begin{equation} \label{eqn:prof_2unpart}
S(p) \triangleq \int_{0}^{\tilde{\theta}} \Big( V-  \theta  \cdot
K(\tilde{Q}, C) \Big) \cdot f(\theta ){\sf d}\theta 
\end{equation}
and the respective provider profit becomes 
\begin{equation} \label{eqn:soc_2unpart}
\pi(p) \triangleq p \tilde{Q} 
\end{equation}

Next, we consider resource partitioning into two service classes. Let $C_1$ and $C_2$ be the respective capacity of each service class, where $C_1 + C_2 = C$. Let $\theta_1$ and $\theta_2$ be the cut-off users at the equilibrium of each service class, and the respective usage be $Q_1 \triangleq F(\theta_1) - F(\theta_2)$ and $Q_2 \triangleq F(\theta_2)$. Because of an identical pricing of each service class at $p$, the social welfare for the partitioned service classes is\footnote{In trying to reduce notations, we are going to slightly abuse our use of them here. The social welfare for the case of a single class service with one price $p$, and the case of multiple service classes with price $\mbox{\boldmath $p$}$ will both be denoted by a function $S(.)$, with a different number of price parameters as appropriate.}: 
\begin{equation} \label{eqn:soc_2part}
\begin{array}{rl}
S(p,p) \triangleq & \displaystyle \int_{0}^{\theta_2} \Big( V-  \theta  \cdot K(Q_2, C_2) \Big) \cdot f(\theta ){\sf d}\theta + \int_{\theta_2}^{\theta_1} \Big( V-  \theta  \cdot K(Q_1, C_1) \Big) \cdot f(\theta ){\sf d}\theta
\end{array} 
\end{equation}
while the respective provider profit is\footnote{Again, the same abuse of notation is applied to the profit function.}: 
\begin{equation} \label{eqn:prof_2part}
\pi(p,p) \triangleq p (Q_1 + Q_2) 
\end{equation}
We next provide a general sufficient condition on $K(Q,C)$ for the viability of service partitioning.

\begin{theorem} \label{thm:part}
For all $0 \le \alpha < 1$ we have the following.
\begin{enumerate}

\item ({\em Partition-preferred congestion function}):
If $K(Q, C) \ge K(\alpha Q, \alpha C)$, 
\begin{equation}
S(p,p) \ge S(p) \mbox{\ and\ } \pi(p,p) \ge \pi(p) 
\end{equation}

\item ({\em Multiplexing-preferred congestion function}):
If $K(Q, C) \le K(\alpha Q, \alpha C)$, 
\begin{equation}
S(p,p) \le S(p) \mbox{\ and\ } \pi(p,p) \le \pi(p) 
\end{equation}

\end{enumerate}
\end{theorem}

The proofs are all in the Appendix unless otherwise stated.

Theorem~\ref{thm:part} classifies two types of congestion functions. Intuitively, the congestion function favoring service partitioning is one that sees decreased congestion externality as we scale down the usage and capacity. Alternatively, the congestion function that favors multiplexing (service merging) is one that sees decreased congestion externality as we scale up the usage and capacity. Note that a congestion function $K(Q, C)$ is given by definition increasing in usage $Q$, but decreasing in capacity $C$. When we scale up both usage and capacity, one or the other of these factors is more dominating, giving rise to the two classes of congestion functions. 

We next apply Theorem~\ref{thm:part} to several specific examples.

\begin{example} 

\begin{enumerate}

\item {\em Utilization-sensitive service:} $K_{\sf utl}(Q, C)$ is \emph{indifferent} to partitioning or multiplexing. 
\begin{equation} 
K_{\sf utl}(Q, C) = \frac{Q}{C} = \frac{\alpha Q}{\alpha C} = K_{\sf
cs}(\alpha Q, \alpha C) 
\end{equation}

\item {\em Latency-sensitive service:} $K_{\sf lat}(Q, C)$ prefers multiplexing to partitioning. 
\begin{equation}
K_{\sf lat}(Q, C) = \frac{1}{C - Q} < \frac{1}{\alpha (C - Q)} = K_{\sf lat}(\alpha Q, \alpha C)
\end{equation}
Further, motivated by Pollaczek-Khinchine formula for M/G/1 queue, we also consider general latency as 
\begin{equation}
K_{\sf glat}(Q, C) = \frac{Q(1+\delta_S^2)}{2C(C-Q)}+\frac{1}{C}, 
\end{equation}
where $\delta_S^2$ is the coefficient of variation of service time, by convention. Similarly, $K_{\sf glat}(Q, C)$ prefers multiplexing to partitioning. 
\begin{equation}
K_{\sf glat}(Q, C) < \frac{1}{\alpha} K_{\sf glat}(Q, C) = K_{\sf glat}(\alpha Q, \alpha C) 
\end{equation}

\item {\em Loss-sensitive service:} Similar to $K_{\sf utl}(Q, C)$,  $K_{\sf los}(Q, C)$ is \emph{indifferent} to partitioning or multiplexing.

\item {\em Outage-sensitive service:} $K_{\sf out}(Q, C)$ prefers partitioning to multiplexing: 
\begin{equation}
K_{\sf out}(Q, C) = (\frac{\epsilon \alpha Q}{\alpha C})^C  > (\frac{\epsilon \alpha Q}{\alpha C})^{\alpha C} = K_{\sf out}(\alpha Q, \alpha C) 
\end{equation}

\end{enumerate}
By means of Theorem~\ref{thm:part}, we conclude that service partitioning is viable for utilization-, loss- and outage-sensitive services. This has been validated in Sec.~\ref{sec:numerical}. Note that Theorem~\ref{thm:part} is sufficiently general and can be applied to diverse types of congestion functions. 
\end{example}
\

The following corollary extends Theorem~\ref{thm:part} from two classes to multiple classes.

\begin{corollary} \label{corpart}
Suppose $0 \le \alpha < 1$, and $\mbox{\boldmath $p$} = (p_i = p)_{i=1}^{m}$, then the following hold.
\begin{enumerate}

\item ({\em Partition-preferred congestion function}):

If $K(Q, C) \ge K(\alpha Q, \alpha C)$ for all $\alpha$, 
\begin{equation}
S(\mbox{\boldmath $p$}) \ge S(p) \mbox{\ and\ } \pi(\mbox{\boldmath $p$}) \ge \pi(p) 
\end{equation}

\item ({\em Multiplexing-preferred congestion function}):

If $K(Q, C) \le K(\alpha Q, \alpha C)$ for all $\alpha$, 
\begin{equation}
S(\mbox{\boldmath $p$}) \le S(p) \mbox{\ and\ } \pi(\mbox{\boldmath $p$}) \le \pi(p) 
\end{equation}
\end{enumerate}
\end{corollary}
\begin{proof}
This is straightforward to prove by induction through proving $m$ service classes is true when supposing $(m-1)$ service classes is true. 
\end{proof}

\subsection{Viability of Differentiated Pricing} \label{sec:different}

Although Sec.~\ref{sec:partition} provides the sufficient condition for service partitioning under identical pricing, it does not cover the case of differentiated pricing. In particular, PMP is not perceivable by users under identical pricing. To complete the picture, in this section we compare differentiated pricing and identical pricing. We will rely on the following property of congestion functions.

\begin{definition} \label{def:mono} ({\em Monotone Preference to Service Classes}) Given a fixed set $\{C_i: i = 1,...,m \}$, the set of congestion functions $\{K(Q_i, C_i): i = 1,...,m \}$ are subject to:
\begin{enumerate}

\item each $K(Q_i, C_i)$ must be strictly increasing and differentiable in $Q_i$, hence, the partial derivative of $K(Q_i, C_i)$ at $Q_i$: $k(Q_i, C_i) \triangleq \frac{\partial K(Q, C)}{\partial Q}|_{Q=Q_i, C=C_i}$  exists and is positive;

\item suppose $C_1 < C_2 < \dots < C_m$, then either one of the following two cases must be true:
\begin{enumerate}

\item[(${\sf m.1}$)] $Q_{i} > Q_{j}$ implying $k(Q_{i},C_{i}) > k(Q_{j},C_{j})$ for any distinct pair $i \ne j$; or

\item[(${\sf m.2}$)] $Q_{i} > Q_{j}$ implying $k(Q_{i},C_{i}) < k(Q_{j},C_{j})$ for any distinct pair $i \ne j$.

\end{enumerate}

\end{enumerate}

\end{definition}

The first condition in Definition~\ref{def:mono} ensures the smoothness of congestion functions, whereas the second condition requires a monotone order on the derivatives of the $m$ service classes. 
Intuitively, a monotone order of the derivatives (i.e., (${\sf m.1}$) or (${\sf m.2}$)) reflects a monotone order of sensitivity of negative externality among the service classes. 
Consequently, such a monotone preference is an indication that, once a user joins the second class, any marginal change to usage in each service class will only shift the user's selection to either the first class or third class (if available). This precludes the switching among service classes that is out of a linear order.

Note that homogeneous service classes with the same convex congestion function (i.e., $C_i = C_j$) obviously satisfy Definition~\ref{def:mono}. Here, we also allow distinct $C_i$ to capture the heterogeneous capacities among different service classes. The provider profit for homogeneous service classes had been studied in \cite{chander89mono}. Here we generalize the result to consider heterogeneous service classes and social welfare. We next apply Definition~\ref{def:mono} to some examples.

\begin{example} 

\begin{enumerate}

\item {\em Utilization-sensitive service:} Suppose $C_1 < C_2 < \dots < C_m$, $K_{\sf utl}(Q, C)$ satisfies monotone preference to service classes. Then 
\begin{eqnarray}
& & k_{\sf utl}(Q_i, C_i) = \frac{\partial K_{\sf utl}(Q, C)}{\partial Q}|_{Q=Q_i, C=C_i} = \frac{1}{C_i} \\
& \Rightarrow & k_{\sf utl}(Q_1, C_1) > k_{\sf utl}(Q_2, C_2) >
\dots > k_{\sf utl}(Q_m, C_m) \mbox{\ for all\ } Q_{1}, Q_{2}, ..., Q_{m}. 
\end{eqnarray} 

\item {\em Latency-sensitive service:} However, $K_{\sf lat}(Q, C)$ does {\em not} always satisfy monotone preference to service classes because 
\begin{equation}
k_{\sf lat}(Q_i, C_i) = \frac{\partial K_{\sf lat}(Q, C)}{\partial Q}|_{Q=Q_i, C=C_i} = \frac{1}{(C_i - Q_i)^2} 
\end{equation}
Note that $\frac{1}{(C_i - Q_i)^2}$ is not a monotone function in $C_i$ or $Q_i$. For instance, when $C_1 = 0.3, C_2 = 0.7, Q_1 = 0.2, Q_2 = 0.5$, then 
$\frac{1}{(C_1 - Q_1)^2} > \frac{1}{(C_2 - Q_2)^2}$
However, when $Q_1 = 0.05, Q_2 = 0.5$, then
$\frac{1}{(C_1 - Q_1)^2} < \frac{1}{(C_2 - Q_2)^2}$.

\item {\em Loss-sensitive service:} Similar to $K_{\sf utl}(Q, C)$,  $K_{\sf los}(Q, C)$ satisfies monotone preference to service classes.

\item {\em Outage-sensitive service:} One also can show that $K_{\sf out}(Q, C)$ satisfies monotone preference to service classes.

\end{enumerate}
 
\end{example}
\

In the following, we compare the maximum social welfare gained by identical pricing and differentiated pricing, given a monotone preference to service classes. Hence, we conclude that differentiated pricing is viable for utilization-, loss- and outage-sensitive services. This has been validated in Sec.~\ref{sec:numerical}. 

\begin{theorem} \label{thm:soc}
Given two service classes that satisfy monotone preference (Definition~\ref{def:mono}), the social welfare obtained through identical pricing at $p$ is strictly inferior to differentiated pricing for some $p_1 \ne p_2$: 
\begin{equation}
S(p_1,p_2) > S(p,p) 
\end{equation}
\end{theorem}

The proof relies on the notion of total derivative $d S(p_1,p_2)$, which can be found in the Appendix (Sec.~\ref{sec:tol_der}).
One might not be surprised to see that the differentiated pricing ($p_1 \ne p_2$) could be better than identical pricing ($p_1 = p_2$), but what is remarkable is the strict superiority of differentiated pricing. In Sec.~\ref{sec:numerical}, we have validated this result for utilization-sensitive service. Note that for a congestion function that does not satisfy monotone preference, the maximum social welfare gained by identical pricing may or may not be higher than that of differentiated pricing. This can be observed in latency-sensitive service. 
Theorem~\ref{thm:soc} can be extended to the case of $m$ service classes as follows. Hence, for utilization-sensitive service, it is advantageous to offer as many service classes as possible, in terms of an increase of social welfare.

\begin{corollary} \label{cor:soc}
Given $m$ service classes that satisfy monotone preference (Definition~\ref{def:mono}), let $\mbox{\boldmath $p$} \triangleq (p_i = p)_{i=1}^{m}$, then there exists $\mbox{\boldmath $p'$} \triangleq (p'_i)_{i=1}^{m}$, such that $p'_i \ne p'_{j}$ for all distinct $i, j \in \{1, ..., m\}$, 
\begin{equation}
S(\mbox{\boldmath $p'$}) > S(\mbox{\boldmath $p$}) 
\end{equation}
\end{corollary}

We can also prove the same results for provider profit.

\begin{theorem} \label{thm:prof}
Given two service classes that satisfy monotone preference (Definition~\ref{def:mono}), the provider profit of identical pricing at $p$ is strictly inferior to that of differentiated pricing for some $p_1 \ne p_2$:  
\begin{equation}
\pi(p_1,p_2) > \pi(p,p) 
\end{equation}
\end{theorem}

\begin{corollary} \label{cor:prof}
Given $m$ service classes that satisfy monotone preference (Definition~\ref{def:mono}), let $\mbox{\boldmath $p$} \triangleq (p_i = p)_{i=1}^{m}$, then there exists $\mbox{\boldmath $p'$} \triangleq (p'_i)_{i=1}^{m}$, such that $p'_i \ne p'_{j}$ for some distinct $i, j \in \{1, ..., m\}$, 
\begin{equation}
\pi(\mbox{\boldmath $p'$}) > \pi(\mbox{\boldmath $p$}) 
\end{equation}
\end{corollary}

Corollary~\ref{cor:prof} is an immediate consequence of Theorem~\ref{thm:prof}, but is weaker than Corollary~\ref{cor:soc}, since provider profit is more difficult to analyze than social welfare. 
\\

{\bf Ramification}: We arrive at the point where we can provide a more complete answer to our original question on the viability of PMP. Combining Theorems~\ref{thm:part}-\ref{thm:prof}, we have a sufficient condition to guarantee PMP to be viable in the sense of both provider profit as well as social welfare. This also explains not only why PMP is viable for utilization-, loss- and outage-sensitive services, but also for latency-sensitive services is not always viable. Theorems~\ref{thm:part}-\ref{thm:prof} are sufficiently general so that they can also be applied to general digital services.

\section{Duopoly Case} \label{sec:duo} 

The study involving multiple competitive providers is more challenging. A plausible outcome is that the providers settle at a Nash equilibrium at which unilateral change in pricing or capacity allocation will induce an inferior profit. 
Previously, \cite{gibbens00JSAC} studied the properties of Nash equilibrium of PMP in the presence of two competitive providers and derived a closed-form solution of Nash equilibrium for the simple setting of utilization services of the same capacity and {\em disallowing user opt-out}. They reported that PMP is not viable in the setting of two competitive providers as compared to the case of simply a single service class offered by each provider.

However, in a more general model {\em with user opt-out} as studied in this article, the viability of PMP is a consequence of more subtle properties of the congestion function. In fact, we find PMP viable in several settings of partition-preferred congestion functions.

In this section, we first derive necessary conditions for Nash equilibrium for two competitive providers. We restrict our analysis to the setting wherein one provider (${\sf I}$) offers a single unpartitioned service class, whereas another provider (${\sf II}$) can flexibly partition its service classes and offer differentiated pricing. The necessary conditions are then applied to the numerical study of several specific congestion functions, from which we compare the viability of PMP. 

\subsection{Characterization of Nash Equilibrium}

We define the profit of provider ${\sf I}$ as 
$\pi^{\sf I} \triangleq p^{\sf I} Q^{\sf I}$
where $p^{\sf I}$ is the price offered for its single service class and $Q^{\sf I}$ is the amount of users accessing it. From Eqn.~(\ref{eqn:c2}), $Q^{\sf I} = Q_i$ if $p^{\sf I}$ is the $i$-th highest price among all the service classes offered.

Provider ${\sf II}$, by contrast, has an option to offer two service classes (offering $p^{\sf II}_{1} \ge p^{\sf II}_{2}$ at capacity $C^{\sf II}_{1}$ and $C^{\sf II}_{2}$ respectively), or a single service class (offering $p^{\sf II}$ at $C^{\sf II} = C^{\sf II}_{1} + C^{\sf II}_{2}$). Hence 
\begin{equation}
\pi^{\sf II} \triangleq
\left\{
\begin{array}{ll}
 p^{\sf II}_1 Q^{\sf II}_1 + p^{\sf II}_2 Q^{\sf II}_2 & \mbox{for two service classes}\\
 p^{\sf II} Q^{\sf II} & \mbox{for one service class}\\
 \end{array}
 \right.
\end{equation}
where $Q^{\sf II}_1, Q^{\sf II}_2, Q^{\sf II}$ are the amount of users accessing the respective service classes.
 
The necessary condition for Nash equilibrium is that derivative $\frac{{\sf d} \pi^{\sf I}}{{\sf d} p^{\sf I}} = 0$, at given $p^{\sf II}_{1}, p^{\sf II}_{2}, C^{\sf II}_{1}, C^{\sf II}_{2}$ (or $p^{\sf II}, C^{\sf II}$).
The respective derivatives are listed for all five cases in Table~\ref{tab:necconds} in the Appendix. In the table, we let $K_i = K(Q_i, C_i)$ and $k_i = \frac{\partial K(Q, C)}{\partial Q}|_{Q=Q_i, C=C_i}$, where $Q_i$ is the total amount of users and $C_i$ is the capacity of $i$-th service class. 

Based on the derivatives, we then numerically evaluate $p^{\sf I}$ such that $\frac{{\sf d} \pi^{\sf I}}{{\sf d} p^{\sf I}} = 0$ at given $C^{\sf I}, p^{\sf II}_{1}, p^{\sf II}_{2}, C^{\sf II}_{1}, C^{\sf II}_{2}$ (or $p^{\sf II}, C^{\sf II}$). Furthermore, we corroborate the existence of Nash equilibrium at the corresponding $(p^{\sf I}, C^{\sf I}, p^{\sf II}_{1}, p^{\sf II}_{2}, C^{\sf II}_{1}, C^{\sf II}_{2})$ and $(p^{\sf I}, C^{\sf I}, p^{\sf II}, C^{\sf II})$ by examining whether $\pi^{\sf I}$ and $\pi^{\sf II}$ are local maxima in the corresponding neighborhood region.

\subsection{Numerical Study and Observations}

Using the preceding results, we particularly study the viability of PMP for utilization-sensitive service. We generalize the congestion function by incorporating default consumption. Define a modified congestion function as  
\begin{equation}
K_{\sf utl.d}(Q, C) \triangleq \frac{Q - \varepsilon}{C}  
\end{equation}
where $\varepsilon \le Q$ is a certain consumption incurred whenever the service class is accessed. This is useful to model the scenario wherein certain overhead or default consumption is imposed in the service.

It is straightforward to show that congestion function $K_{\sf utl.d}(Q, C)$ strictly prefers service partitioning to multiplexing (when $\varepsilon > 0$), and satisfies monotone preferences to service classes.

Figs.~\ref{fig:duo-e00}-\ref{fig:duo-V3} show the profit of each provider at various given values of $p^{\sf I}$, and the corresponding best response from provider ${\sf II}$ on its pricing and capacity so as to maximize its profit $\pi^{\sf II}$.

{\bf Observations}:
We obtain the following key observations from Figs.~\ref{fig:duo-e00}-\ref{fig:duo-V3}:
\begin{enumerate}

\item PMP is viable in several settings of utilization service with and without default consumption. Because provider ${\sf II}$ always has a higher profit when it partitions its service classes (i.e., lines of ``2 vs 1'' in the figures), so is the profit of provider ${\sf I}$.

\item Multiplexing-preferred congestion functions ($K_{\sf utl.d}(Q, C)$) produce higher viability for PMP.

\item When more capacity is given to provider ${\sf II}$, the benefit of PMP is more prominent (see the larger gap between profits of provider ${\sf I}$ and provider ${\sf II}$ in Fig.~\ref{fig:duo-V3}).

\end{enumerate}

We note that although general analytical results are harder for the competitive case, our numerical results agree with the observations in other multiplexing-preferred congestion functions.

\begin{figure*}[htb!] 
\centering 
  \begin{minipage}[c]{.47\textwidth} \centering
    \includegraphics[scale=0.9]{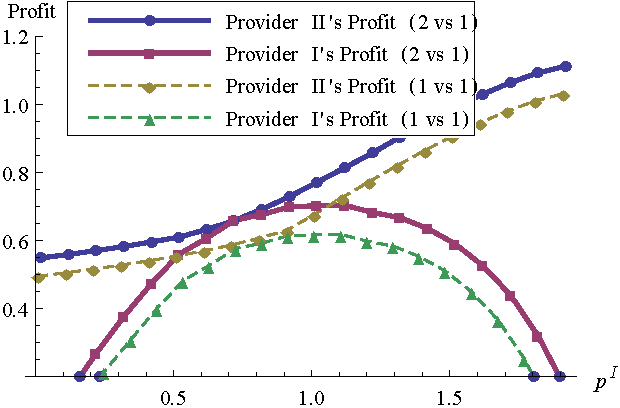}  
    \caption{Utilization-sensitive service ($K_{\sf utl}(Q, C), V = 2, C^{\sf I} = C^{\sf II} = 1$).} \label{fig:duo-e00}
  \end{minipage}
  \hfill \quad
  \begin{minipage}[c]{.47\textwidth} \centering
    \includegraphics[scale=0.9]{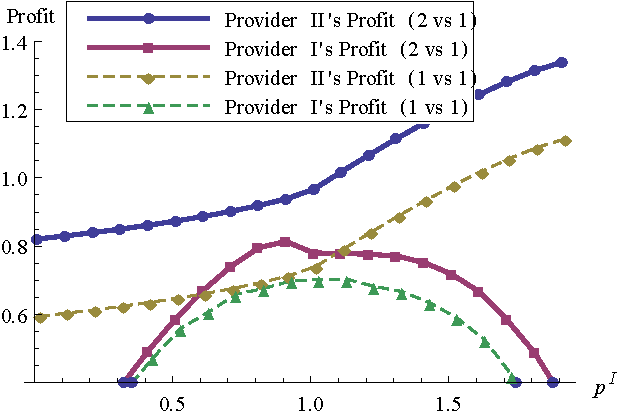} 
    \caption{Utilization-sensitive service with default consumption ($K_{\sf utl.d}(Q, C), V = 2, C^{\sf I} = C^{\sf II} = 1, \varepsilon = 0.1$).} \label{fig:duo-e01}
  \end{minipage} 
\end{figure*}

\begin{figure*}[htb!] 
\centering 
  \begin{minipage}[c]{.47\textwidth} \centering
    \includegraphics[scale=0.9]{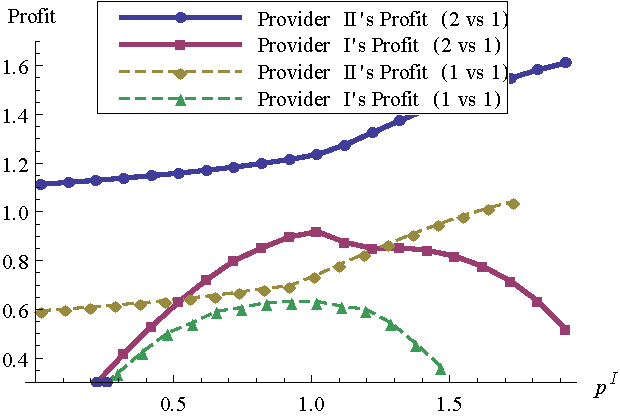}  
    \caption{Utilization-sensitive service with default consumption ($K_{\sf utl.d}(Q, C), V = 2, C^{\sf I} = C^{\sf II} = 1, \varepsilon = 0.2$).} \label{fig:duo-e02}
  \end{minipage}
  \hfill \quad
  \begin{minipage}[c]{.47\textwidth} \centering
    \includegraphics[scale=0.9]{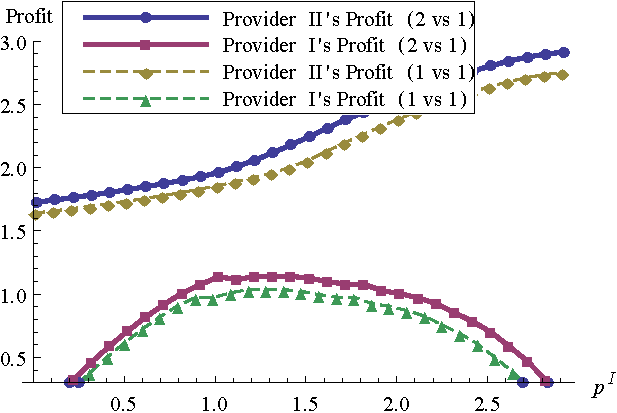} 
    \caption{Utilization-sensitive service ($K_{\sf utl}(Q, C), V = 3, C^{\sf I}, C^{\sf II} = 2$).} \label{fig:duo-V3}
  \end{minipage} 
\end{figure*}

\section{Conclusion and Discussion} \label{sec:disc} 

This article provides general conditions for the viability of Paris Metro Pricing (PMP)
based on a general setting of negative externality (via a general congestion function)
that can model a wide range of digital services. There are two separate messages here,
both of which can be intuitively stated. The first one says that one’s service either
prefers multiplexing (having more people share proportionally more capacity) or not;
if one wants to guarantee a gain (in terms of profit or social welfare) by dividing one’s
service into multiple classes with the same price, then this service would better not
prefer multiplexing. The second message says that, if one starts with a multiclass
service with the same price and one wants to move to charging different prices, then
the service classes one sets up would better generate a monotone linear preference
perceived by the users. By combining these two rules together, we can characterize
a large class of services that can benefit from PMP. These observations also help us
understand why sometimes PMP is not viable. Our results help clarify the confusion
caused by conflicting results on the viability of PMP by previous studies. Our model
is general and the results are rigorously proved and applicable to future studies on
network economics.

\bibliographystyle{plain}
\bibliography{paperbib}

\section{Appendix}

\subsection{Proof for Theorem \ref{thm:part}}
\begin{proof}
First, by Eqn.~(\ref{eqn:c3}) in Definition~\ref{def:equ}, considering a single service class, we obtain 
\begin{equation} \label{eqn:pf1:theta0}
p = V - \tilde{\theta} \cdot K\big(\tilde{Q} , C\big) \ \Rightarrow\
\tilde{\theta} = \frac{V-p}{ K\big(\tilde{Q} , C\big)} 
\end{equation}
where $\tilde{\theta}$ is the cut-off user for a single service class.
When there are two service classes after resource partitioning, we obtain 
\begin{equation} \label{eqn:prf1:eq1}
\left\{
\begin{array}{@{}r@{}l}
p - p = & \theta_2 \cdot\Big( K\big(Q_2 , C_2\big) - K\big(Q_1 , C_1\big) \Big) \\
p = & V - \theta_1 \cdot K\big(Q_1 , C_1\big) 
\end{array}
\right. 
\end{equation}
Hence due to Eqn.~(\ref{eqn:prf1:eq1}) it follows that 
\begin{eqnarray}
K\big(Q_2 , C_2\big) = & K\big(Q_1 , C_1\big) \label{eqn:pf1:K} \\
\theta_1 = & \displaystyle \frac{V - p}{ K(Q_1 , C_1)} \label{eqn:pf1:theta1} 
\end{eqnarray}
Using Eqns.~(\ref{eqn:pf1:theta0}) and (\ref{eqn:pf1:theta1}), we obtain the following equality. 
\begin{equation} \label{eqn:pf1:ratio}
\frac{\tilde{\theta}}{\theta_1} = \frac{K(Q_1, C_1)}{K(\tilde{Q}, C)} = \frac{K(Q_2, C_2)}{K(\tilde{Q}, C)} 
\end{equation}

To complete the proof, we proceed in two steps as follows.

{\sf Step 1}: Then, we want to show:
\begin{enumerate}
\item If $K(Q, C) \ge K(\alpha Q, \alpha C)$ for all $0 \le \alpha \le 1$, then $\theta_1 \ge \tilde{\theta}$

\item If $K(Q, C) \le K(\alpha Q, \alpha C)$ for all $0 \le \alpha \le 1$, then $\theta_1 \le \tilde{\theta}$

\end{enumerate}
Without loss of generality, we assume $C_1 \ge C_2$. Let $\beta \triangleq \frac{C_1}{C}$, and hence we have $\beta \ge 1-\beta$ (i.e., $1 \ge \frac{1 - \beta}{\beta}$).

First, we consider partition-preferred congestion function $K(Q, C) \ge K(\alpha Q, \alpha C)$ for all $0 \le \alpha \le 1$. By Eqn.~(\ref{eqn:pf1:K}), it follows that 
\begin{equation} \label{eqn:pf-1}
K(Q_2,C_2) = K(Q_1,C_1) \ge K(\frac{1-\beta}{\beta} Q_1,C_2) 
\end{equation}
Because $K(Q, C)$ is increasing in $Q$, we obtain 
\begin{equation} \label{eqn:pf1:bd}
Q_2 \ge \frac{1-\beta}{\beta} Q_1 \ \Rightarrow\ Q_1 + Q_2 \ge \frac{Q_1}{\beta} 
\end{equation}
Because $K(Q, C) \ge K(\alpha Q, \alpha C)$ for all $0 \le \alpha \le 1$, we obtain 
\begin{equation} \label{eqn:pf1:effect1}
K(\tilde{Q}, C) \ge K\big(\beta \tilde{Q}, \beta C \big) = K\big(\beta \tilde{Q}, C_1\big) 
\end{equation}
We next use contradiction to show that $\theta_1 \ge \tilde{\theta}$. On the contrary, we suppose $\tilde{\theta} > \theta_1$. Then, by Eqns.~(\ref{eqn:pf1:ratio}) and (\ref{eqn:pf1:effect1}), we obtain 
\begin{equation} 
K\big(Q_1, C_1\big) > K\big(\tilde{Q}, C\big) \ge K\big(\beta \tilde{Q}, C_1\big) 
\end{equation}
Hence, $Q_1 > \beta \tilde{Q}$ because $K(Q, C)$ is increasing in
$Q$.
Also, we note that 
\begin{equation}
\tilde{\theta} > \theta_1\ \Rightarrow\ F(\tilde{\theta}) > F(\theta_1)\ \Rightarrow\ \tilde{Q} > Q_1 + Q_2 
\end{equation}
Therefore we derive 
\begin{equation} \label{eqn:pf-2}
Q_1 > \beta \tilde{Q} > \beta (Q_1 + Q_2) 
\end{equation}
which is a contradiction to Eqn.~(\ref{eqn:pf1:bd}). Hence it should be $\theta_1 \ge \tilde{\theta}$. 

For multiplexing, to show that $K(Q, C) \le K(\alpha Q, \alpha C)$ for all $\alpha$ implies $\theta_1 \le \tilde{\theta}$, we note that we can reverse the signs ``$\ge$'' and ``$>$'' in Eqns~(\ref{eqn:pf-1})-(\ref{eqn:pf-2}) to ``$\le$'' and ``$<$'' respectively.

{\sf Step 2}: Next, we want to show:
\begin{enumerate}
\item if $\theta_1 \ge \tilde{\theta}$, then $S(p,p) \ge S(p)$ and $\pi(p,p) \ge \pi(p)$; and

\item if $\theta_1 \le \tilde{\theta}$, then $S(p,p) \ge S(p)$ and $\pi(p,p) \le \pi(p)$.

\end{enumerate}

The case of provider profit $\pi(p,p)$ and $\pi(p)$ follows from
Eqns.~(\ref{eqn:prof_2unpart}) and (\ref{eqn:prof_2part}) and 
\begin{eqnarray}
\theta_1 \ge \tilde{\theta} & \Rightarrow \ Q_1 + Q_2 \ge \tilde{Q} \ \Rightarrow \ \pi(p,p)\ge \pi(p) \label{eqn:pf1:prof} \\ 
\theta_1 \le \tilde{\theta} & \Rightarrow \ Q_1 + Q_2 \le \tilde{Q} \ \Rightarrow \ \pi(p,p)\le \pi(p) \label{eqn:pf1:ineq1} 
\end{eqnarray} 
where the service classes are charged at an identical price.

For the case of social welfare $S(p,p)$ and $S(p)$, by Eqn.~(\ref{eqn:pf1:K}), we obtain 
\begin{equation}
S(p,p) = \int_{0}^{\theta_1}
\Big( V-  \theta  \cdot K(Q_2, C_2) \Big) \cdot f(\theta ){\sf d}\theta 
\end{equation}
From Eqn.~(\ref{eqn:pf1:ratio}), we obtain 
\begin{eqnarray} 
\theta_1 \ge \tilde{\theta} & \Rightarrow  K\big(\tilde{Q}, C\big) \ge K\big(Q_2, C_2\big) \ \Rightarrow S(p,p)\ge S(p) \label{eqn:pf1:soc}  \\
\theta_1 \le \tilde{\theta} & \Rightarrow K\big(\tilde{Q}, C\big) \le K\big(Q_2, C_2\big)  \ \Rightarrow S(p,p)\le S(p)  \label{eqn:pf1:ineq2} 
\end{eqnarray}
which follows from Eqns.~(\ref{eqn:soc_2unpart}) and (\ref{eqn:pf1:ineq1})-(\ref{eqn:pf1:ineq2}).
\end{proof}

\subsection{An Overview of Total Derivative} \label{sec:tol_der}

Before we proceed to the proof of Theorem \ref{thm:soc}, we briefly revisit the notion of total derivative \cite{EngMaths} that will be useful in the following proofs. 

The marginal change of real function $f(x, y, z)$ (i.e., derivative $d f$) with respect to parameters $x, y, z$ can be written as 
\begin{equation}
d f = \frac{\partial f}{\partial x} d x + \frac{\partial f}{\partial y} d y + \frac{\partial f}{\partial z} d z 
\end{equation}
where $\frac{\partial f}{\partial x}, \frac{\partial f}{\partial y}, \frac{\partial f}{\partial z}$ are the partial derivatives of $f$ at the respective parameter while keeping other parameters as constants. Note that $\frac{\partial f}{\partial x}, \frac{\partial f}{\partial y}, \frac{\partial f}{\partial z}$ are also functions of $x, y, z$.

The derivative $d f$ can be regarded as a function of $(d x, d y, d z)$, each of them representing the marginal change of parameters $x, y, z$. Also, $(d x, d y, d z)$ can be regarded as a vector in the 3D Euclidean space.
An immediate consequence is that if function $f$ is a stationary point at some $(x_0, y_0, z_0)$ (e.g., the maximum), then $d f = 0$ at $(x_0, y_0, z_0)$ irrespective of what values of $(d x, d y, d z)$ we pick. Otherwise, we will be possible to pick a vector $(d x, d y, d z)$, such that $d f \ne 0$ implying that $f$ is not a stationary point at $(x_0, y_0, z_0)$.

\subsection{Proof for Theorem \ref{thm:soc}}

\begin{proof} First, we write $S = S(p_1,p_2)$. We study how the total derivative of $S$ (see Sec.~\ref{sec:tol_der}), $dS$, changes at $p_1 = p_2$. From Eqn.~(\ref{eqn:soc_2part}), 
\begin{equation}
\begin{array}{@{}r@{\ }l@{\ }l}
S = & \displaystyle V \cdot F(\theta_1)  - K(Q_2, C_2) \int_{0}^{\theta_2} {\theta \cdot f(\theta ){\sf d}\theta } - K(Q_1, C_1) \int_{\theta_2}^{\theta_1} {\theta  \cdot f(\theta ){\sf d}\theta }
\end{array} 
\end{equation}
Recall that $k(Q_i, C_i) \triangleq \frac{\partial K(Q, C)}{\partial Q}|_{Q=Q_i, C=C_i}$, $Q_2 \triangleq F(\theta_2)$ and $Q_1 \triangleq F(\theta_1) - F(\theta_2)$.

Note that an equilibrium can be characterized by tuple $(p_1, p_2)$, a pair of independent variables. Similarly, an equilibrium can also be equivalently characterized by tuple $(\theta_1, \theta_2)$, by solving Eqn.~(\ref{eqn:c3}) in Definition~\ref{def:equ}. Then $(\theta_1, \theta_2)$ are treated as a pair of independent variables. Thus,
we take the total derivative of $S$ with respect to $(d \theta_1, d \theta_2)$, and obtain the following: 
\begin{equation} 
\begin{array}{@{}l}
 dS  = \\ 
 \displaystyle V \cdot f(\theta_1)d\theta_1 - \Big( \int_{0}^{\theta_2} {\theta \cdot f(\theta ){\sf d}\theta } \Big) \cdot k(Q_2, C_2) f(\theta_2 )d\theta_2  - K(Q_2, C_2) \cdot \theta_2 \cdot f(\theta_2 )d\theta_2  \\
   \displaystyle - \Big( \int_{\theta_2 }^{\theta_1 } {\theta  \cdot f(\theta ){\sf d}\theta } \Big) \cdot k(Q_1, C_1) \Big( f(\theta_1 )d\theta_1  - f(\theta_2 )d\theta_2 \Big)   - K(Q_1, C_1) \cdot \Big( \theta_1  \cdot f(\theta_1 )d\theta_1  - \theta_2  \cdot f(\theta_2 )d\theta_2 \Big)
\end{array} 
\end{equation}

Then, at identical pricing $p_1 = p_2$, we have $K(Q_1, C_1) = K(Q_2, C_2)$. Hence, we obtain 
\begin{equation} \hspace{-30pt}
\begin{array}{@{}l}
 dS|_{p_1 = p_2} = \\
 \displaystyle V \cdot f(\theta_1)d\theta_1 - \Big( \int_{0}^{\theta_2} {\theta \cdot f(\theta ){\sf d}\theta } \Big) \cdot k(Q_2, C_2) f(\theta_2 )d\theta_2  \\
   - \Big( \int_{\theta_2 }^{\theta_1 } {\theta  \cdot f(\theta ){\sf d}\theta } \Big) \cdot k(Q_1, C_1) \Big( f(\theta_1 )d\theta_1  - f(\theta_2 )d\theta_2 \Big)  - K(Q_1, C_1) \cdot \theta_1  \cdot f(\theta_1 )d\theta_1
\end{array} 
\end{equation}
Next, we pick a vector $(d\theta_1, d\theta_2)$, and show that $dS|_{p_1 = p_2}$ will strictly increase in the direction of $(d\theta_1, d\theta_2)$. Such a vector indeed exists if we keep $\theta_1$ as a constant (i.e., $d\theta_1 = 0$). 

First, we obtain 
\begin{subequations}\label{eqn:pf2:ineq2} 
  \begin{align} 
& dS|_{p_1 = p_2, d\theta_1 = 0} \\
 =\  &  \bigg( \Big( \int_{\theta_2 }^{\theta_1 } {\theta  \cdot f(\theta ){\sf d}\theta } \Big) \cdot k(Q_1, C_1) \notag - \Big( \int_{0}^{\theta_2} {\theta \cdot f(\theta ){\sf d}\theta } \Big) \cdot k(Q_2, C_2)\bigg) f(\theta_2 )d\theta_2 \\
 >\ &  \theta_2 \cdot \bigg( \Big( \int_{\theta_2 }^{\theta_1 } {f(\theta ){\sf d}\theta } \Big) \cdot k(Q_1, C_1) \notag  - \Big( \int_{0}^{\theta_2} {f(\theta ){\sf d}\theta } \Big) \cdot k(Q_2, C_2)\bigg) f(\theta_2 )d\theta_2 \\
=\ & \theta_2 \Big(Q_1 \cdot k(Q_1, C_1) - Q_2 \cdot k(Q_2, C_2)\Big) f(\theta_2 )d\theta_2
  \end{align}  
\end{subequations}  
Without loss of generality, we assume $Q_1 \ge Q_2$. Since the two service classes satisfy monotone preference (Definition~\ref{def:mono}), if we always pick $d\theta_2 > 0$ in the case of (${\sf m.1}$), and $d\theta_2 < 0$ in the case of (${\sf m.2}$), then it is always true that the total derivative $dS|_{p_1 = p_2, d\theta_1 = 0} > 0$.
Therefore, we see that the social welfare $S$ can strictly increase by differentiated pricing ($p_1 \ne p_2$) from identical pricing ($p_1 = p_2$).
\end{proof}

\subsection{Proof for Corollary \ref{cor:soc}}

\begin{proof}
By Theorem~\ref{thm:soc}, it is true for $m=2$.  For $m=3$, it
follows that $p_1 = p_2 = p_3$ cannot be optimal. Next, we also show
that $p_1 > p_2 = p _3$ and $p_1 = p_2 > p _3$ cannot be optimal.
Then, the total derivative of social welfare $S$ with three service classes becomes 
\begin{equation} 
\begin{array}{@{}l}
 dS  =  \\
 \displaystyle V \cdot f(\theta_1)d\theta_1 - \Big( \int_{0}^{\theta_3} {\theta \cdot f(\theta ){\sf d}\theta } \Big) \cdot k(Q_3, C_3) f(\theta_3 )d\theta_3  - K(Q_3, C_3) \cdot \theta_3 \cdot f(\theta_3 )d\theta_3  \\
  \displaystyle - \Big( \int_{\theta_3 }^{\theta_2 } {\theta  \cdot f(\theta ){\sf d}\theta } \Big) \cdot k(Q_2, C_2) \Big( f(\theta_2 )d\theta_2  - f(\theta_3 )d\theta_3 \Big)   - K(Q_2, C_2) \cdot \Big( \theta_2  \cdot f(\theta_2 )d\theta_2  - \theta_3  \cdot f(\theta_3 )d\theta_3 \Big) \\
   \displaystyle - \Big( \int_{\theta_2 }^{\theta_1 } {\theta  \cdot f(\theta ){\sf d}\theta } \Big) \cdot k(Q_1, C_1) \Big( f(\theta_1 )d\theta_1  - f(\theta_2 )d\theta_2 \Big)  - K(Q_1, C_1) \cdot \Big( \theta_1  \cdot f(\theta_1 )d\theta_1  - \theta_2  \cdot f(\theta_2 )d\theta_2 \Big)
\end{array} 
\end{equation}
For $p_1 > p_2 = p _3$ (i.e., $K(Q_2, C_2)=K(Q_3, C_3)$), setting $d
\theta_1 = d \theta_2 = 0$ will degenerate to the case $m=2$.  For
$p_1 = p_2 > p _3$ (i.e., $K(Q_1, C_1)=K(Q_2, C_2)$), setting $d
\theta_1 = d \theta_3 = 0$ will also degenerate to the case $m=2$.
Hence, it follows for $m=3$ is true. Using an iterative argument, we
can show that it is true for all $m \ge 2$.
 \end{proof}
 
\subsection{Proof for Theorem \ref{thm:prof}}

\begin{proof}
Similar to Theorem~\ref{thm:soc}, taking the total derivative of
$\pi$ with respect to $(d p_1, d p_2)$, we have  
\begin{equation} 
\begin{array}{@{}l}
d\pi = 
Q_2 dp_2 + p_2  \cdot f(\theta_2 )\Big(
{\frac{{\partial \theta_2 }}{{\partial p_1 }}dp_1  +
  \frac{{\partial \theta_2 }}{{\partial p_2 }}dp_2 } \Big)  + Q_1 dp_1  \\
   \qquad + p_1  \cdot \bigg(
  f(\theta_1 )\Big( \frac{{\partial \theta_1 }}{{\partial p_1 }}dp_1
    + \frac{{\partial \theta_1 }}{{\partial p_2 }}dp_2  \Big) 
-  f(\theta_2 )\Big( {\frac{{\partial \theta_2 }}{{\partial p_1 }}dp_1  + \frac{{\partial \theta_2 }}{{\partial p_2 }}dp_2 } \Big) \bigg)
\end{array}  
\end{equation}
Then, at identical pricing ($p_1 = p_2$), we obtain  
\begin{equation}
d\pi |_{p_2  = p_1 }  = Q_2 dp_2  + Q_1 dp_1 + p_1  \cdot f(\theta_1
)(\frac{{\partial \theta_1 }}{{\partial p_1 }}dp_1  +
\frac{{\partial \theta_1 }}{{\partial p_2 }}dp_2 )  
\end{equation}
However, $d\pi$ is more difficult than $d S$, involving $\frac{{\partial \theta_1 }}{{\partial p_1 }}$ and $\frac{{\partial \theta_1 }}{{\partial p_2 }}$.

As in Theorem~\ref{thm:soc}, we pick a vector $(dp_1, dp_2)$,
and show that $d\pi|_{p_1 = p_2}$ will strictly increase in the
direction of $(dp_1, dp_2)$. To achieve this we keep $\theta_1$ as a
constant (i.e., $d\theta_1 = \frac{{\partial \theta_1 }}{{\partial
p_1 }}dp_1  + \frac{{\partial \theta_1 }}{{\partial p_2 }}dp_2 =
0$). Hence  
\begin{equation}
d\pi |_{p_1 = p_2, d\theta_1 = 0}  = Q_1 dp_1  + Q_2 dp_2 
\end{equation}
Also, from Eqn.~(\ref{eqn:c3}) in Definition~\ref{def:equ}, we obtain the total
derivatives $(p_1, p_2)$ with respect to $(d p_1, d p_2)$ as 
\begin{equation} \hspace{-30pt} 
\left\{
\begin{array}{@{}r@{\ }l}
dp_1   = & - K (Q_1, C_1)d\theta_1  - \theta_1  \cdot k (Q_1, C_1)( f(\theta_1)d\theta_1  - f(\theta_2)d\theta_2) \\
dp_1  - dp_2   = &  (K (Q_2, C_2) - K (Q_1, C_1))d\theta_2  \\
&  +  \theta_2  \cdot \Big( k (Q_2, C_2)( f(\theta_2 )d\theta_2 )  - k (Q_1, C_1)( f(\theta_1)d\theta_1  - f(\theta_2)d\theta_2) \Big)
\end{array} 
\right. 
\end{equation} 
We keep $\theta_1$ as a constant (i.e., $d\theta_1 = 0$), and by identical pricing $p_1 = p_2$ $\Rightarrow$ $K(Q_1, C_1) = K(Q_2, C_2)$, we have 
\begin{equation} \label{eqn:prf4:solve}
\left\{
\begin{array}{@{}r@{\ }l@{}}
dp_1 = &  \theta_1  \cdot k(Q_1, C_1)f(\theta_2 ) d\theta_2 \\
dp_1  - dp_2 = &  \theta_2  \cdot \Big(k(Q_2, C_2) +
k(Q_1, C_1)\Big)\cdot f(\theta_2 )  d\theta_2
\end{array}
\right. 
\end{equation}
Solving Eqn.~(\ref{eqn:prf4:solve}) for $(dp_1, dp_2)$, we obtain 
\begin{equation} \label{eqn:pf4:sub}
dp_2  = \frac{ - \theta_2  \cdot (k(Q_2, C_2) + k(Q_1, C_1)) +
\theta_1  \cdot k(Q_1, C_1) }
 {\theta_1  \cdot k(Q_1, C_1)}dp_1 
\end{equation}
Because keeping $\theta_1$ as a constant (i.e., $d\theta_1 = 0$), by
substituting Eqn.~(\ref{eqn:pf4:sub}) we obtain 
\begin{subequations} \label{eqn:pf4:ineq2} \hspace{-30pt}
  \begin{align}  
& d\pi |_{p_1 = p_2, d\theta_1 = 0}  = Q_2 dp_2  + Q_1 dp_1 \\
=\ & \Big( Q_2  \cdot \big(  - \theta_2  \cdot (k(Q_2, C_2) + k(Q_1,
C_1))
+ \theta_1  \cdot k(Q_1, C_1) \big)  + Q_1 \theta_1 \cdot k(Q_1, C_1) \Big)
\frac{dp_1}{\theta_1  \cdot k(Q_1, C_1)} \notag \\
=\ &  \big( - Q_2 \theta_2 \cdot (\frac{k(Q_2, C_2)}{k(Q_1, C_1)}) - Q_2 \theta_2 + Q_2 \theta_1 + Q_1 \theta_1 \big) \frac{dp_1}{\theta_1 }  \\
>\ &  \big( - Q_2 \cdot (\frac{k(Q_2, C_2)}{k(Q_1, C_1)}) - Q_2 + Q_2 + Q_1 \big) \frac{ \theta_2 dp_1}{\theta_1 } \\
= &  \big( Q_1 \cdot k(Q_1, C_1) - Q_2 \cdot k(Q_2,
C_2) \big) \frac{ \theta_2 dp_1}{\theta_1 \cdot k(Q_1, C_1)}
 \end{align}  
\end{subequations}
Without loss of generality, we assume $Q_1 > Q_2$. 
Since the two service classes satisfy monotone preference
(Definition~\ref{def:mono}), if we always pick $d\theta_2 > 0$ in the case
of (${\sf m.1}$) and $d\theta_2 < 0$ in the case of (${\sf m.2}$),
then it is always true that $d\pi|_{p_1 = p_2, d\theta_1 = 0} > 0$.
Therefore, the provider profit $\pi$ can strictly increase by
differentiated pricing ($p_1 \ne p_2$) from identical pricing ($p_1
= p_2$).
\end{proof}

\medskip

\subsection{Derivatives of $\frac{{\sf d} \pi^{\sf I}}{{\sf d} p^{\sf I}}$}

\begin{table}[htb!]  \centering
\caption{Derivatives of all cases of $\frac{{\sf d} \pi^{\sf I}}{{\sf d} p^{\sf I}}$.\label{tab:necconds}}{
\begin{tabular}{@{}c@{ \ }|l@{}} 
\hline\hline
 & \\
 Cases &  $\frac{{\sf d} \pi^{\sf I}}{{\sf d} p^{\sf I}}$\\
  & \\
\hline
 & \\
$p^{\sf I} \ge p^{\sf II}_1 \ge p^{\sf II}_2$ & 
$\frac{1}{k_1 \theta_1} \Big(
-{p^{\sf I}}+{k_1} {Q_1} {\theta_1}$\\ 
& \quad $+\frac{{K_1} {p^{\sf I}} (-{k_2} {k_3} {\theta_2} {\theta_3}+{k_1} {\theta_1} ({K_2}-{K_3}- 2 {k_3} {\theta_3})+({K_1}-{K_2}- 2 {k_2} {\theta_2}) ({K_2}-{K_3}- 2 {k_3} {\theta_3}))}{{k_1} {k_2} {\theta_1} {\theta_2} ({K_2}-{K_3}- 2 {k_3} {\theta_3})-({K_1}+{k_1} {\theta_1}) ({k_2} {k_3} {\theta_2} {\theta_3}-({K_1}-{K_2}- 2 {k_2} {\theta_2}) ({K_2}-{K_3}- 2 {k_3} {\theta_3}))} \Big)$ \\ 
$p^{\sf II}_1 \ge p^{\sf I} \ge p^{\sf II}_2$ & 
${Q_2}+\frac{{p^{\sf I}}}{{K_2}-{K_3}- 2 {k_3} \theta_3}$ \\
& \quad $- \frac{{p^{\sf I}} ({K_1}+{k_1} \theta_1) ({K_2}-{K_3}+{k_2} \theta_2- 2 {k_3} \theta_3) (-{K_2}+{K_3}+{k_3} \theta_3)}{(-{K_2}+{K_3}+ 2 {k_3} \theta_3) ({k_2} {k_3} ({K_1}+{k_1} \theta_1) \theta_2 \theta_3+(-{k_1} {k_2} \theta_1 \theta_2-({K_1}+{k_1} \theta_1) ({K_1}-{K_2}- 2 {k_2} \theta_2)) ({K_2}-{K_3}- 2 {k_3} \theta_3))}$
\\ 
$p^{\sf II}_1 \ge p^{\sf II}_2 \ge p^{\sf I}$ &
${Q_3}+\frac{{p^{\sf I}}}{{K_2}-{K_3}- 2 {k_3} {\theta_3}}$ \\ 
& \quad $+ \frac{{k_2} {k_3} {p^{\sf I}} ({K_1}+{k_1} {\theta_1}) {\theta_2} {\theta_3}}{(-{K_2}+{K_3}+ 2 {k_3} {\theta_3}) ({k_2} {k_3} ({K_1}+{k_1} {\theta_1}) {\theta_2} {\theta_3}+(-({K_1}-{K_2}) ({K_1}+{k_1} {\theta_1})+{k_2} (_2 {K_1}+{k_1} {\theta_1}) {\theta_2}) ({K_2}-{K_3}- 2 {k_3} {\theta_3}))}$\\ 
 & \\
\hline
 & \\
$p^{\sf I} \ge p^{\sf II}$ &  
$\frac{{K_1}^2 {Q_1}+{K_2} ({p^{\sf I}}-{Q_1} ({K_1}+{k_1} {\theta_1}))+{k_2} ({p^{\sf I}}-{k_1} {Q_1} {\theta_1}) {\theta_2}+{K_1} {Q_1} ({k_1} {\theta_1}- 2 {k_2} {\theta_2})}{{K_1}^2-{k_1} {\theta_1} ({K_2}+{k_2} {\theta_2})-{K_1} ({K_2}-{k_1} {\theta_1}+ 2 {k_2} {\theta_2})}$ \\ 
$p^{\sf II} \ge p^{\sf I}$ & 
${Q_2}-\frac{p^{\sf I} ({K_1}+{k_1} {\theta_1})}{-{k_1} {k_2} {\theta_1} {\theta_2}-({K_1}+{k_1} {\theta_1}) ({K_1}-{K_2}-_2 {k_2} {\theta_2})}$ \\
 & \\
\hline\hline
\end{tabular} 
}
\end{table}

\end{document}